\documentclass{lmcs}
\pdfoutput=1

\usepackage{lastpage}
\lmcsdoi{19}{4}{35}
\lmcsheading{}{\pageref{LastPage}}{}{}%
{Oct.~28,~2022}{Dec.~21,~2023}{}

\usepackage[utf8]{inputenc}
\usepackage{amsmath,amssymb,microtype}
\usepackage{thm-restate}
\usepackage{url}
\usepackage{tabto}
\usepackage{multirow}
\usepackage{tikz}
\usepackage[textsize=footnotesize]{todonotes}
\usetikzlibrary{automata,positioning,shapes}

\usepackage{newfile}
\newcommand{\NP}{\mathsf{NP}}

\title{Existential definability over the subword ordering}

\def \subword{\preccurlyeq}
\def \ssubword{\prec}

\def \nats{\mathbb{N}}

\def \RE{\mathsf{RE}}

\def \cT{\mathcal{T}}
\newcommand{\N}{\mathbb{N}}

\newcommand{\FOconst}[1]{(#1^*,{\subword},(w)_{w\in #1^*})}
\newcommand{\FOsomeconst}[2]{(#1^*,{\subword},#2)}
\newcommand{\FOpure}[1]{(#1^*,{\subword})}

\author[P.~Baumann]{Pascal Baumann\lmcsorcid{0000-0002-9371-0807}}
\author[M.~Ganardi]{Moses Ganardi\lmcsorcid{0000-0002-0775-7781}}
\author[R.~S.~Thinniyam]{Ramanathan S. Thinniyam\lmcsorcid{0000-0002-9926-0931}}
\author[G.~Zetzsche]{Georg Zetzsche\lmcsorcid{0000-0002-6421-4388}}

\address{Max Planck Institute for Software Systems (MPI-SWS), Kaiserslautern, Germany}
\email{\{pbaumann, ganardi, thinniyam, georg\}@mpi-sws.org}

\begin{document}

\begin{abstract}
  We study first-order logic (FO) over the structure consisting of finite words
  over some alphabet~$A$, together with the (non-contiguous) subword ordering.
  In terms of decidability of quantifier alternation fragments, this logic is
  well-understood: If every word is available as a constant, then even the
  $\Sigma_1$ (i.e., existential) fragment is undecidable, already for binary
  alphabets $A$.
  
  However, up to now, little is known about the expressiveness of the
  quantifier alternation fragments: For example, the undecidability proof for
  the existential fragment relies on Diophantine equations and only
  shows that recursively enumerable languages over a singleton alphabet (and
  some auxiliary predicates) are definable.
  
  We show that if $|A|\ge 3$, then a relation is definable in the existential
  fragment over $A$ with constants if and only if it is recursively enumerable.
  This implies characterizations for all fragments~$\Sigma_i$: If $|A|\ge 3$,
  then a relation is definable in $\Sigma_i$ if and only if it belongs to the
  $i$-th level of the arithmetical hierarchy.  In addition, our result yields
  an analogous complete description of the $\Sigma_i$-fragments for $i\ge 2$ of
  the \emph{pure logic}, where the words of $A^*$ are not available as
  constants.
\end{abstract}

\maketitle

\section{Introduction}
\label{sec:introduction}

\subsection*{The subword ordering} 
A word $u$ is a \emph{subword} of another word $v$ if $u$ can be obtained from $v$ by deleting letters at an arbitrary set of positions.
The subword ordering has been studied intensively over the last few decades. On
the one hand, it appears in many classical results of theoretical computer science. For example, subwords have been a central topic in string
algorithms~\cite{baeza1991,elzinga2008algorithms,maier1978complexity}. Moreover, their combinatorial properties are the basis for verifying lossy channel systems~\cite{abdulla1996verifying}.
Particularly in recent years, subwords have received a considerable amount of attention.
Notable examples include lower bounds in fine-grained complexity~\cite{BringmannK15,BringmannK18}, algorithms to compute the set of all subwords of formal languages~\cite{DBLP:conf/lics/AtigCHKSZ16,DBLP:conf/mfcs/AtigMMS17,DBLP:conf/icalp/BarozziniCCP20,DBLP:conf/lics/ClementePSW16,DBLP:conf/icalp/HabermehlMW10,DBLP:conf/popl/HagueKO16,DBLP:conf/icalp/Zetzsche15,DBLP:conf/stacs/Zetzsche15,Zetzsche2016a,DBLP:conf/lics/Zetzsche18,Goubault-Larrecq2020,DBLP:conf/concur/AnandZ23}, and applications thereof to infinite-state verification~\cite{DBLP:journals/corr/abs-1111-1011,DBLP:conf/concur/TorreMW15,DBLP:journals/pacmpl/BaumannMTZ22,DBLP:journals/lmcs/MajumdarTZ22,DBLP:conf/icalp/BaumannMTZ20,DBLP:conf/icalp/0001GMTZ23,DBLP:conf/icalp/0001GMTZ23a}.
Subwords are also the basis of Simon's congruence~\cite{SakarovitchSimon1997}, which has recently been studied
from algorithmic~\cite{FleischerK18,GawrychowskiKKM21,DBLP:conf/rp/FleischmannKKMNSW23} and combinatorial~\cite{BarkerFHMN20,DayFKKMS21,KarandikarKS15,KarandikarS19,DBLP:conf/cwords/SchnoebelenV23} viewpoints.

\subsection*{First-order logic over subwords}
The importance of subwords has motivated the study of first-order logics (FO) over
the subword ordering. This has been considered in two variants: In the \emph{pure logic}, one has FO over the structure $\FOpure{A}$, where $A$ is an alphabet and $\subword$ is the subword ordering. In the version \emph{with constants}, we have the structure $\FOconst{A}$, which has a constant for each word from $A^*$.
Traditionally for FO, the primary questions are \emph{decidability} and \emph{definability}, particularly regarding quantifier alternation fragments~$\Sigma_i$. Here, decidability refers to the \emph{truth problem}: 
Given a formula $\varphi$ in a particular fragment over $\FOpure{A}$ or $\FOconst{A}$, respectively, does  $\varphi$ hold? 
By definability, we mean understanding which relations can be defined by formulas in a particular fragment. \linebreak
The $\Sigma_i$-fragment consists of formulas in prenex form that begin with existential quantifiers and then alternate $i-1$ times between blocks of universal and existential quantifiers. For example, the formula 
\[ \exists x\colon (a\not\subword x~\vee~b\not\subword x)~\wedge~x\not\subword u~\wedge~x\subword v \]
belongs to the $\Sigma_1$-fragment, also called the \emph{existential fragment} over $\FOconst{A}$ with $A=\{a,b\}$. The formula has free variables $u,v$ and refers to the constants $a$ and $b$. It holds if and only if $v$ has more $b$'s or more $a$'s than $u$.

For FO over subwords, decidability is well-understood. 
In the pure logic, the $\Sigma_2$-fragment is undecidable, already over two letters~\cite[Corollary~III.6]{HalfonSZ17}, whereas the $\Sigma_1$-fragment (i.e., existential formulas) is decidable~\cite[Theorem 2.2]{kuske2006theories} and $\NP$-complete~\cite[Theorem 2.1]{KarandikarS15}. 
This fueled hope that the $\Sigma_1$-fragment might even be decidable with constants, but this turned out to be undecidable, already over two letters~\cite[Theorem III.3]{HalfonSZ17}. Decidability (and complexity) have also been studied for the two-variable fragment~\cite{KarandikarS15,KarandikarS19,kuske_et_al:LIPIcs:2020:12729}, and extended with counting quantifiers and regular predicates~\cite{kuske_et_al:LIPIcs:2020:12729,KuskeZ19}.

Nevertheless, little is known about definability. Kudinov, Selivanov, and
Yartseva have shown that using arbitrary first-order formulas over
$\FOpure{A}$, one can define exactly the relations from the arithmetical
hierarchy\footnote{Also known as the Kleene–Mostowski hierarchy}
that are invariant under automorphisms of
$\FOpure{A}$~\cite[Theorem 5]{kudinov2010definability}, if $|A|\ge 2$.
However, this does not explain definability of the $\Sigma_i$-fragments. For
example, in order to define all recursively enumerable languages, as far as we can see, their proof
requires several quantifier alternations. An undecidability proof by
Karandikar and Schnoebelen~\cite[Theorem 4.6]{KarandikarS15} for the
$\Sigma_2$-fragment can easily be adapted to show that for each alphabet $A$,
there exists a larger alphabet $B$ such that every recursively enumerable
language $L\subseteq A^*$ is definable in the $\Sigma_2$-fragment over
$\FOconst{B}$. However, a full description of the expressiveness of the
$\Sigma_2$-fragment is missing.

\subsection*{Existential formulas}
The expressiveness of existential formulas is even further from being
understood.  The undecidability proof in \cite{HalfonSZ17} reduces from
solvability of Diophantine equations, i.e., polynomial equations over integers,
which is a well-known undecidable problem~\cite{matiyasevich1993hilbert}. To this end, it
is shown in \cite{HalfonSZ17} that the relations $\mathsf{ADD}=\{(a^m,a^n,a^{m+n}) \mid m,n\in\N\}$ and
$\mathsf{MULT}=\{(a^m,a^n,a^{m\cdot n}) \mid m,n\in\N\}$ are definable existentially
using the subword ordering, if one has at least two letters. Since Diophantine equations can be used to
define all recursively enumerable relations over natural numbers, this implies
that all recursively enumerable relations involving a single letter are definable
existentially. However, this says little about which languages (let alone relations)
over more than one letter are definable. For example, it is not clear whether the language of all $w\in\{a,b\}^*$
that do \emph{not} contain $aba$ as an infix, or the reversal relation $\mathsf{REV}_A=\{(u,v) \mid u,v\in A^*,~\text{$v$ is the reversal of $u$}\}$, 
are definable---it seems particularly difficult to define them over the subword ordering using the methods from \cite{HalfonSZ17}.

\subsection*{Contribution}
We show that for any alphabet $A$ with $|A| \geq 3$, every recursively
enumerable relation $R \subseteq (A^*)^k$, $k\in\N$, is existentially definable in
$\FOconst{A}$. In fact, similarly to an observation made in \cite{HalfonSZ17},
we even show that there is a single sufficiently complex word $W \in A^*$ such that
the structure $\FOsomeconst{A}{W}$ with just this single constant symbol suffices to
existentially define all recursively enumerable relations.
Since every existentially definable relation is clearly
recursively enumerable (via a simple enumerative algorithm), this completely
describes the expressiveness of existential formulas for $|A|\ge 3$. 
Despite the undecidability of the existential fragment~\cite{HalfonSZ17}, we find it surprising that all recursively enumerable relations---including relations like $\mathsf{REV}_A$---are existentially definable. 

Our result yields characterizations of the $\Sigma_i$-fragments for every~$i\ge
2$:  It implies that for each $i\ge 2$, the $\Sigma_i$-fragment over
$\FOconst{A}$ can define exactly the relations in $\Sigma_i^0$, the $i$-th
level of the arithmetical hierarchy, assuming $|A|\ge 3$. This also
provides a description of $\Sigma_i$ in the pure logic: It follows that in the
$\Sigma_i$-fragment over $\FOpure{A}$, one can define exactly the relations in
$\Sigma_i^0$ that are invariant under automorphisms of $\FOpure{A}$, if
$|A|\ge 3$.

Since \cite{HalfonSZ17} shows that all recursively enumerable languages over one letter are definable in $\FOconst{A}$ if $|A|\ge 2$, it would suffice to define a bijection between $a^*$ and $A^*$ using subwords. 
However, since this seems hard to do directly, our proof follows a different route. 
We first show how to define rational transductions and then a special language  from which one can build every recursively enumerable relation via rational transductions and intersections. 
In particular, a byproduct is a direct proof of undecidability of the existential fragment in the case of $|A|\ge 3$ that avoids using undecidability of Diophantine equations\footnote{Our proof relies on the definability of concatenation and certain counting predicates (see \autoref{sec:proof}), which was shown directly in~\cite{HalfonSZ17}, without using computational completeness of Diophantine equations.}.

\subsection*{Key ingredients} 
The undecidability proof for the existential fragment from~\cite{HalfonSZ17}
shows that the relations $\mathsf{ADD}$ and $\mathsf{MULT}$ are definable,
in addition to auxiliary predicates that are needed for this, such as
concatenation and letter counting predicates of the form ``$|u|_a=|v|_b$''.
With these methods, it is difficult to express that a certain property holds
locally---by which we mean: at every position in a word.  Using concatenation,
we can define languages like~$(a^nb)^*$ for each $n\in\N$ (see \autoref{sec:proof}), which ``locally
look like $a^nb$''. But if we want to express that, e.g., $aba$ does not occur as an
infix, this is of little help, because words avoiding an infix need not be
periodic. The ability to disallow infixes would aid us in defining rational transductions via runs of
transducers, as these are little more than configuration sequences where pairs of configurations
that are not connected by a transition do not occur as infixes.
Such local properties are often easy to state with universal
quantification, but this is not available in existential formulas.

An important theme in our proof is to express such local properties
by carefully constructing long words in which $w$ has to embed in order for $w$
to have the local property. For example, our first lemma
says: Each set $X\subseteq A^{=\ell}$ can be characterized as the set of
words (of length $\ge \ell$) that embed into each word in a finite set
$P$. This allows us to define sets $X^*$.

Steps~I--III of our proof use techniques of this type to express rational transductions. In Step~IV, we then define the special language
$G=\{a^nb^n\mid n\ge 0\}^*$, which has the property that all recursively
enumerable languages can be obtained from $G$ using rational transductions and
intersection. This yields all recursively enumerable relations
over two letters in Step~V.

In sum, Steps~I--V let us define all recursively enumerable relations over
$\{a,b\}$, provided that the alphabet $A$ contains an additional auxiliary
letter. It then remains to define recursively enumerable relations that can
also involve all other letters in $A$. We do this in Step~VI by observing that
each word $w\in A^*$ is determined by its projections to binary alphabets
$B\subseteq A$.  This allows us to compare words by looking at two letters at a
time and use the other (currently unused) letters for auxiliary means.

A conference version of this paper appeared in~\cite{BaumannGTZ22}.

\section{Main results}
\label{sec:results}
We say that $u$ is a {\em subword} of $v$, written $u \subword v$,
if there exist words $u_1, \dots, u_n$ and $v_0, \dots, v_n$ such that
$u = u_1 \cdots u_n$ and $v = v_0 u_1 v_1 \cdots u_n v_n$.

\subsection*{Subword logic}
We consider first-order logic over the structure $\FOpure{A}$,
first-order logic over the structure $\FOsomeconst{A}{w_1,\ldots,w_n}$
enriched with finitely many constant symbols $w_1, \ldots, w_n \in A^*$,
and first-order logic over the structure $\FOconst{A}$
enriched with constant symbols $w$ for every word $w \in A^*$.
A first-order formula $\varphi$ with free variables $x_1, \dots, x_k$ {\em defines} a relation $R \subseteq (A^*)^k$ if $R$ contains exactly those tuples of words $(v_1, \dots, v_k)$ that satisfy\footnote{The correspondence between the entries in the tuple and the free variables of $\varphi$ will always be clear, because the variables will have an obvious linear order by sorting them alphabetically and by their index. For example, if $\varphi$ has free variables $x_i$ for $1\le i\le k$ and $y_j$ for $1\le j\le\ell$, then we order them as $x_1,\ldots,x_k,y_1,\ldots,y_\ell$.} the formula $\varphi$.

Let us define the quantifier alternation fragments of first-order logic.
A formula without quantifiers is called \emph{$\Sigma_0$-formula} or \emph{$\Pi_0$-formula}.
For $i\ge 1$, a $\Sigma_i$-formula (resp.\ $\Pi_i$-formula) is one of the form $\exists x_1\cdots \exists x_n \varphi$ (resp.\ $\forall x_1\cdots \forall x_n\varphi$), where $\varphi$ is a a $\Pi_{i-1}$-formula (resp.\ $\Sigma_{i-1}$-formula), $x_1,\ldots,x_n$ are variables, and $n\ge 0$.
In other words, a $\Sigma_i$-formula is in prenex form and its quantifiers begin with a block of existential quantifiers and alternate at most $i-1$ times between universal and existential quantifiers. The \emph{$\Sigma_i$-fragment} (\emph{$\Pi_i$-fragment}) consists of the $\Sigma_i$-formulas ($\Pi_i$-formulas).
In particular, the $\Sigma_1$-fragment (called the \emph{existential fragment}) consists of the formulas in prenex form that only contain existential quantifiers.

\subsection*{Expressiveness with constants}
Our main technical contribution is the following.
\begin{thm}\label{thm:main}
	Let $A$ be an alphabet with $|A|\ge 3$. A relation is definable in the $\Sigma_1$-fragment over $\FOconst{A}$ if and only if it is recursively enumerable.
\end{thm}
\noindent We prove \autoref{thm:main} in \autoref{sec:proof}.
\autoref{thm:main} in particular yields a description of what is expressible using
$\Sigma_i$-formulas for each $i\ge 1$. Recall that the \emph{arithmetical
hierarchy} consists of classes $\Sigma_1^0,\Sigma_2^0,\ldots$, where
$\Sigma_1^0 = \RE$ is the class of recursively enumerable relations, and for $i\ge
2$, we have $\Sigma_i^0=\RE^{\Sigma_{i-1}^0}$. Here, for a class of relations
$\mathcal{C}$, $\RE^{\mathcal{C}}$ denotes the class of relations recognized by
oracle Turing machines with access to oracles over the class $\mathcal{C}$.
\begin{restatable}{cor}{levelByLevelConstants}\label{level-by-level-constants}
Let $A$ be an alphabet with $|A|\ge 3$ and let $i\ge 1$. A relation is definable in the $\Sigma_i$-fragment over $\FOconst{A}$ if and only if it belongs to $\Sigma_i^0$.
\end{restatable}

\noindent By \cite[Theorem~3.5]{HalfonSZ17arxiv} the undecidability of the $\Sigma_1$-fragment already holds for
$\FOsomeconst{A}{W}$ where $W \in A^*$ is a sufficiently complex constant.
Using the same ideas we show that the characterizations from \autoref{thm:main} and \autoref{level-by-level-constants}
also already hold for a single constant,
which will be proven in \autoref{sec:single-constant}.

\begin{rem}\label{rem:single-constant}
	Let $|A| \ge 3$. There exists a word  $W \in A^*$ so that \autoref{thm:main} and \autoref{level-by-level-constants}
	still hold for $\FOsomeconst{A}{W}$ instead of $\FOconst{A}$.
\end{rem}

\subsection*{Expressiveness of the pure logic} 
\autoref{level-by-level-constants} completely describes the relations definable in
the structure $\FOconst{A}$ if $|A|\ge 3$.  We can use this to derive a
description of the relations definable without constants, i.e., in the structure
$\FOpure{A}$.  The lack of constants slightly reduces the expressiveness; to
make this precise, we need some terminology. An \emph{automorphism (of
	$\FOpure{A}$)} is a bijection $\alpha\colon A^*\to A^*$ such that
	$u\subword v$ if and only if $\alpha(u)\subword\alpha(v)$.  A relation
	$R\subseteq (A^*)^k$ is \emph{automorphism-invariant} if for every
	automorphism $\alpha$, we have $(v_1,\ldots,v_k)\in R$ if and only if
	$(\alpha(v_1),\ldots,\alpha(v_k)) \in R$.  It is straightforward to check
	that every formula over $\FOpure{A}$ defines an automorphism-invariant
	relation.  Thus, in the $\Sigma_i$-fragment over $\FOpure{A}$, we can
	only define automorphism-invariant relations inside $\Sigma_i^0$.

\begin{restatable}{cor}{levelByLevelPure}\label{level-by-level-pure}
Let $A$ be an alphabet with $|A|\ge 3$ and let $i\ge 2$. A relation is definable in the $\Sigma_i$-fragment
over $\FOpure{A}$ if and only if it is automorphism-invariant and
belongs to $\Sigma_i^0$.
\end{restatable}

\noindent To give some intuition on automorphism-invariant sets, let us recall the
classification of automorphisms of $\FOpure{A}$,
shown implicitly by Kudinov, Selivanov, and Yartseva in
\cite{kudinov2010definability} (for a short and explicit proof,
see~\cite[Lemma~3.8]{HalfonSZ17arxiv}):
A map $\alpha\colon A^*\to A^*$ is an automorphism of $\FOpure{A}$ if and only if
(i) the restriction of $\alpha$ to $A$ is a permutation of $A$, and
(ii) $\alpha$ is either a \emph{word morphism}, i.e., $\alpha(a_1\cdots a_k)=\alpha(a_1)\cdots\alpha(a_k)$ for any $a_1,\ldots,a_k\in A$,
or a \emph{word anti-morphism}, i.e., $\alpha(a_1\cdots a_k)=\alpha(a_k)\cdots\alpha(a_1)$ for any $a_1,\ldots,a_k\in A$.

Finally, \autoref{level-by-level-pure} raises the question of whether the
$\Sigma_1$-fragment over $\FOpure{A}$ also expresses exactly the
automorphism-invariant recursively enumerable relations. It does not:
\begin{obs}\label{sigma1-pure}
Let $|A|\ge 2$. There are undecidable binary relations definable in the $\Sigma_1$-fragment
over $\FOpure{A}$. However, not every automorphism-invariant regular
language is definable in it.
\end{obs}

\section{Existentially defining recursively enumerable relations}\label{sec:proof}
In this section, we prove \autoref{thm:main}. Therefore, we now concentrate on
definability in the $\Sigma_1$-fragment. Moreover, for an alphabet $A$, we will sometimes use the phrase
\emph{$\Sigma_1$-definable over $A$} as a shorthand for definability in the $\Sigma_1$-fragment over the structure $\FOconst{A}$.

\subsection*{Notation}
For an alphabet $A$, we write $A^{=k}$, $A^{\ge k}$, and $A^{\le k}$ for the
set of words over $A$ that have length exactly $k$, at least $k$, and at most
$k$, respectively.
We write $|w|$ for the length of a word $w$.
If $B \subseteq A$ is a subalphabet of $A$ then $|w|_B$ denotes the number of occurrences of letters $a \in B$ in $w$,
or simply $|w|_a$ if $B = \{a\}$ is a singleton.
Furthermore, we write $\pi_B \colon A^* \to B^*$ for the projection morphism
which keeps only the letters from $B$. 
If $B=\{a,b\}$, we also write $\pi_{a,b}$ for $\pi_{\{a,b\}}$.
The {\em downward closure} of a word $v \in A^*$ is defined as $v {\downarrow} :=  \{u \in A^* \mid u \subword v\}$.

\subsection*{Basic relations}

We will use two kinds of relations, concatenation and counting letters, which are shown to be
$\Sigma_1$-definable in $\FOconst{A}$ as part of the undecidability proof of
the truth problem in~\cite[Theorem III.3]{HalfonSZ17}. The following relations
are $\Sigma_1$-definable if $|A|\ge 2$.
\begin{description}
	\item[Concatenation] The relation $\{(u,v,w)\in (A^*)^3 \mid w = uv \}$.
	\item[Counting letters] The relation $\{(u,v)\in (A^*)^2 \mid |u|_a=|v|_b\}$ for any $a,b\in A$.
\end{description}
Moreover, we will make use of a classical fact from word combinatorics: For
$u,v\in A^*$, we have $uv=vu$ if and only if there is a word $r\in A^*$ with
$u\in r^*$ and $v\in r^*$~\cite{Berstel1979}. In particular, if $p$ is
\emph{primitive}, meaning that $p\in A^+$ and there is no $r\in A^*$ with
$|r|<|p|$ and $p\in r^*$, then $up=pu$ is equivalent to $u\in p^*$. 
Furthermore, note that by counting letters as above,
and using concatenation, we can also say $|u|_a = |vw|_a$,
i.e., $|u|_a=|v|_a+|w|_a$ for $a\in A$.
With these building blocks, we can state arbitrary linear equations over terms
$|u|_a$ with $u\in A^*$ and $a\in A$. For example, we can say
$|u|=3\cdot |v|_a+2\cdot |w|_b$ for $u,v,w\in A^*$ and $a,b\in A$.
This also allows us to state modulo constraints, such as
$\exists v\colon |u|_a = 2 \cdot |v|_a$, i.e., ``$|u|_a$ is even''.
Finally, counting letters lets us define projections: Note that for $B\subseteq
A$ and $u,v\in A^*$, we have $v=\pi_B(u)$ if and only if $v\subword u$ and
$|v|_b=|u|_b$ for each $b\in B$ as well as $\neg (a \subword v)$ for every $a \in A \setminus B$.

For any subalphabet $B \subseteq A$ one can clearly define $B^*$ over $A$.
Hence definability of a relation over $B$ also implies definability over the larger alphabet $A$.

\subsection*{Finite state transducers} An important ingredient of our proof
is to define regular languages in the subword order,
and, more generally, rational transductions,
i.e., relations recognized by finite state transducers.

For $k \in \nats$, a \emph{$k$-ary finite state transducer} 
$\cT = (Q,A,\delta,q_0,Q_f)$ consists of a finite set of \emph{states} $Q$,
an input alphabet $A$, an \emph{initial state} $q_0 \in Q$, a set of \emph{final states} 
$Q_f \subseteq Q$, and a \emph{transition relation} $\delta \subseteq Q \times (A \cup \{\varepsilon\})^k \times Q$. 
For a \emph{transition} $(q,a_1,\ldots,a_k,q') \in \delta$, we also write 
$q \xrightarrow{(a_1,\ldots,a_k)} q'$.

The transducer $\cT$ \emph{recognizes} the $k$-ary relation
$R(\cT) \subseteq (A^*)^k$
containing precisely those $k$-tuples $(w_1,\ldots,w_k)$, for which there is a transition sequence
$q_0 \xrightarrow{(a_{1,1},\ldots,a_{k,1})} q_1 \xrightarrow{(a_{1,2},\ldots,a_{k,2})} \ldots \xrightarrow{(a_{1,m},\ldots,a_{k,m})} q_m$
with $q_m \in Q_f$ and $w_i = a_{i,1}a_{i,2} \cdots a_{i,m}$ for
all $i \in \{1, \ldots, k\}$. Such a transition sequence is called an
\emph{accepting run} of $\cT$. We sometimes prefer to think of the $w_i$ as \emph{produced output} rather than \emph{consumed input} and thus occasionally use terminology accordingly. 
A relation $T$ is called a \emph{rational transduction} if it is recognized by some finite state transducer~$\cT$. Unary transducers (i.e., $k=1$) recognize the \emph{regular~languages}.

\subsection*{Overview}
As outlined in the introduction, our proof consists of six steps. In
Steps~I--III, we show that we can define all rational transductions $T\subseteq
(A^*)^k$ over the alphabet $B$, if $|B|\ge |A|+1$. In Step IV, we define the
special language $G=\{a^nb^n\mid n\ge 0\}^*$. From~$G$, all recursively
enumerable languages can be obtained using rational transductions and
intersection, which in Step~V allows us to define over $B$ all recursively
enumerable relations over $A$, provided that $|B|\ge |A|+1$. Finally, in
Step~VI, we use projections to binary alphabets to define arbitrary recursively
enumerable relations over $A$, if $|A|\ge 3$.

\subsection*{Step I: Defining Kleene stars}
We first define the languages $X^*$, where $X$ consists of words of
equal length. To this end, we establish an alternative representation for such sets.

\begin{exa}
Before proving the general statement,
let us illustrate how to define the language $\{ab,ba\}^*$ using an auxiliary symbol $\#$.
It suffices to define the language $\{ab\#,ba\#\}^*$ and then project to $\{a,b\}$.
The simple but key observation is that
\begin{equation}
	\label{abba-example-1}
	u \in \{ab,ba\} \iff u \subword bab \text{ and } u \subword aba \text{ and } |u| \ge 2.
\end{equation}
We claim that a word $w$ belongs to $\{ab\#,ba\#\}^*$ if and only if
\begin{equation}
	\label{abba-example-2}
	\exists n \in \N \colon w \subword (aba \#)^n \wedge w \subword (bab \#)^n \wedge |w|_\# = n \wedge |w| = 3n.
\end{equation}
The ``only if''-direction is immediate. For the ``if''-direction consider a word $w$ satisfying 
\eqref{abba-example-2}, i.e.\ $w = w_1 \# \cdots w_n \#$ where each $w_i$ belongs to $\{a,b\}^*$.
Then each word $w_i$ is a subword of $aba$ and $bab$.
Since $|w| = 3n$ either all words $w_i$ have length $2$ and therefore $w_i \in \{ab,ba\}$ by~\eqref{abba-example-1};
or, there exists some $w_i$ with $|w_i| > 2$.
But then again $w_i \in \{ab,ba\}$ by~\eqref{abba-example-1}, contradicting $|w_i| > 2$.
\end{exa}

\begin{lem}
	\label{lem:x-rep}
	Every nonempty set $X \subseteq A^{=\ell}$ can be written as $X = A^{\ge \ell} \cap \bigcap_{p \in P} p {\downarrow}$
	for some finite set $P \subseteq A^*$.
\end{lem}

\begin{proof}
	We can assume $\ell \ge 1$
	since otherwise $X = \{\varepsilon\} = A^{\ge 0} \cap \varepsilon {\downarrow}$.
	Let $w \in A^*$ be any permutation of $A$ (i.e., each letter of $A$ appears exactly once in $w$).
	If $a \in A$, then $(w \setminus a)$ denotes the word obtained from $w$ by deleting $a$.
	For any nonempty word $u = a_1 \cdots a_k \in A^+$, $a_1,\ldots,a_k\in A$, define the word
	\[
		p_u = (w \setminus a_1) \, (w \setminus a_1) \, a_1 \, (w \setminus a_2) \, (w \setminus a_2) \, a_2 \cdots (w \setminus a_{k-1}) \, (w \setminus a_{k-1}) \, a_{k-1} \, (w \setminus a_k) \, (w \setminus a_k).
	\]
	Note that $p_u$ does not contain $u$ as a subword:
  In trying to embed each letter $a_i$ of $u$ into $p_u$, the first possible choice for $a_1$ comes after the initial sequence $(w \setminus a_1) \, (w \setminus a_1)$.
  Similarly, the next possible choice for each subsequent $a_i$ is right after $(w \setminus a_i) \, (w \setminus a_i)$.
  However, this only works until $a_{k-1}$, since there is no $a_k$ at the end of $p_u$.
  
	On the other hand, observe that $p_u$ contains every word $v \in A^{\le k} \setminus \{u\}$ as a subword:
	Suppose that $v = b_1 \cdots b_m$, $b_1,\ldots,b_m\in A$, and let $i \in [1,m+1]$ be the minimal position with $b_i \neq a_i$ or $i = m + 1$.
	The prefix $b_1 \cdots b_{i-1} = a_1 \cdots a_{i-1}$ occurs as a subword of $p_u$, which in the case $i = m + 1$ already is the whole word $v$. If $i \le m$ then $b_i$ occurs in $(w \setminus a_i)$, and $b_{i+1} \cdots b_m$ embeds into
	the subword $(w \setminus a_i) \, a_i \cdots (w \setminus a_{k-1}) \, a_{k-1}$ of $p_u$.
	Thus, we can write
	\[
		X ~ = A^{\ge \ell} ~\cap~ \bigcap_{u \in (A^{= \ell} \setminus X) \cup A^{= \ell+1}} p_u {\downarrow}.
	\]
  Here $u \in A^{= \ell+1}$ was added to also exclude all words of length greater than $\ell$.
\end{proof}

\begin{lem}
	\label{lem:x-star}
	Let $A \subseteq B$ be finite alphabets and $\# \in B \setminus A$.
	Let $X \subseteq A^{=k}$ and $Y \subseteq A^{=\ell}$ be sets.
	Then $(X \# Y \#)^*$ and $X^*$ are $\Sigma_1$-definable over $B$.
\end{lem}

\begin{proof}
We can clearly assume that $X,Y$ are nonempty.
By \autoref{lem:x-rep} we can write $X = A^{\ge k} \cap \bigcap_{p \in P} p {\downarrow}$
and $Y = A^{\ge \ell} \cap \bigcap_{q \in Q} q {\downarrow}$
for some finite sets $P,Q \subseteq A^*$.
We claim that $w \in (A \cup \{\#\})^*$ belongs to $(X \# Y \#)^*$ if and only if 
\begin{equation}
	\label{eq:x-sharp-formula}
	\exists n \in \mathbb{N} \colon |w|_\# = 2n \wedge |w|_A = (k+\ell) \cdot n \wedge \bigwedge_{p \in P, q \in Q} w \subword (p \# q \#)^n.
\end{equation}
Observe that the number $n$ is uniquely determined by $|w|_\#$.
The ``only if''-direction is clear. Conversely, suppose that $w \in (A \cup \{\#\})^*$ satisfies the formula.
We can factorize $w = x_1 \# y_1 \# \dots x_n \# y_n \#$
where each $x_i$ is a subword of each word $p \in P$,
and each $y_i$ is a subword of each word $q \in Q$.
If some word $x_i$ were strictly longer than $k$, then it would belong to $X$ by the representation of $X$,
and in particular would have length $k$, contradiction.
Therefore each word $x_i$ has length at most $k$, and similarly each word $y_i$ has length at most~$\ell$.
However, since the total length of $x_1 y_1 \dots x_n y_n$ is $(k+\ell) \cdot n$,
we must have $|x_i| = k$ and $|y_i| = \ell$, and hence $x_i \in X$ and $y_i \in Y$ for all $i \in [1,n]$. This proves our claim.

Finally, \eqref{eq:x-sharp-formula} is equivalent to the following $\Sigma_1$-formula:
\[
	(k + \ell) \cdot |w|_\# = 2 \cdot |w|_A ~\wedge~ \bigwedge_{p \in P, q \in Q} \exists u \in (p \# q \#)^* \colon (w \subword u ~\wedge~ |u|_\# = |w|_\#)
\]
Here, we express $u\in (p\#q\#)^*$ as follows. If $p\ne q$, then $p\#q\#$ is primitive and $u\in (p\#q\#)^*$ is equivalent to $u(p\#q\#)=(p\#q\#)u$. If $p=q$, then $u\in (p\#q\#)^*$ is equivalent to $up\#=p\#u$ and $|u|_{\#}$ being even. Finally, to define $X^*$ we set $Y = \{\varepsilon\}$ and obtain $X^* = \pi_A((X \# Y \#)^*)$.
\end{proof}

\subsection*{Step II: Blockwise transductions}
On our way towards rational transductions, we work with a subclass of transductions.
If $T \subseteq A^* \times A^*$ is any subset, then we define the relation 
\[ T^* = \{(x_1 \cdots x_n, \,y_1 \cdots y_n) \mid n\in\N,~(x_1,y_1), \dots, (x_n,y_n) \in T \}. \]
We call a transduction {\em blockwise} if it is of the form $T^*$
for some $T \subseteq A^{=k} \times A^{=\ell}$ and $k,\ell \in \N$.
\begin{lem}
	\label{lem:blockwise}
	Let $A \subseteq B$ be finite alphabets with $|B| \ge |A| + 1$.
	Every blockwise transduction $R \subseteq A^* \times A^*$ is $\Sigma_1$-definable over $B$.
\end{lem}
\begin{proof}
Let $\# \in B \setminus A$ be a symbol.
Suppose that $R = T^*$ for some $T \subseteq A^{=k} \times A^{=\ell}$.
Define the language $L = \{x \# y \# \mid (x,y) \in T \}^*$.
Note that
\[
	w \in L \iff w \in (A^{=k} \# A^{=\ell} \#)^* ~\wedge~ \pi_A(w) \in \{ xy \mid (x, y) \in T \}^*,
\]
and hence $L$ is $\Sigma_1$-definable over $B$ by \autoref{lem:x-star}.
The languages $X = (A^{=k} \# \#)^*$ and $Y = (\# A^{=\ell} \#)^*$ are also definable over $B$ by \autoref{lem:x-star}.
Then $(x,y) \in R$ if and only if
\[
	\exists w \in L, \, \hat x \in X, \, \hat y \in Y \colon~ \hat x, \hat y \subword w ~\wedge~ |w|_\# = |\hat x|_\# = |\hat y|_\# ~\wedge~ x = \pi_A(\hat x) ~\wedge~ y = \pi_A(\hat y).\qedhere
\]
\end{proof}

\subsection*{Step III: Rational transductions}
We are ready to define arbitrary rational transductions.

\begin{lem}
	\label{lem:transductions}
	Let $A \subseteq B$ be finite alphabets where $|A| + 1 \le |B|$ and $|B|\ge 3$.
	Every rational transduction $T \subseteq (A^*)^k$ is $\Sigma_1$-definable over $B$.
\end{lem}
\begin{proof}
	\newcommand{\prun}{\mathsf{run}}
	\newcommand{\pinput}{\mathsf{input}}
	Let $a,b\in B$. Let us first give an overview. Suppose the transducer
	for $T$ has $n$ transitions. Of course, we may assume that every run contains at least one transition. The idea is that a sequence of transitions
	is encoded by a word, where transition $j\in\{1,\ldots,n\}$ is
	represented by $a^jb^{n+1-j}$. We will define predicates $\prun$ and $\pinput_i$ for $i\in\{1,\ldots,k\}$ with
	\[(w_1,\ldots,w_k)\in T\iff \exists w\in\{a,b\}^*\colon ~\prun(w) ~\wedge~ \bigwedge_{i=1}^k \pinput_i(w,w_i). \]
	Here, $\prun(w)$ states that $w$ encodes a
	sequence of transitions that is a run of the transducer.  Moreover, $\pinput_i(w,w_i)$ states that $w_i\in A^*$ is the input of this run in the $i$-th coordinate.

	We begin with the predicate $\prun$. Let us call the words in $X=\{a^jb^{n+1-j}\mid
	j\in\{1,\ldots,n\}\}$ the \emph{transition codes}.
	Let $\Delta$ be the set of all words $a^ib^{n+1-i} a^jb^{n+1-j}$ for which the target state of transition $i$ and
	the source state of transition $j$ are the same.
	Note that a word $w\in X^*$ represents a run if
	\begin{enumerate}
		\item $w$ begins with a transition that can be applied in an initial state,
		\item $w$ ends with a transition that leads to a final state, and
		\item either $w \in \Delta^* \cap X \Delta^* X$ or $w \in X\Delta^* \cap \Delta^*X$, depending on whether the run has an even or an odd number of transitions.
	\end{enumerate}
	Thus, we can define $\prun(w)$ using prefix and suffix relations and
	membership to sets~$\Delta^*$. The prefix and suffix relation can be defined over $\{a,b\}$
	using concatenation, see also~\cite[Theorem~III.3, step 14]{HalfonSZ17}.
	Finally, we can express $w\in X^*$, $w \in \Delta^*$ and similar with \autoref{lem:x-star}.

	It remains to define the $\pinput_i$ predicate.
	In the case that every transition reads a single letter on each input (i.e., no $\varepsilon$ input),
	we can simply replace each transition code in $w$ by its $i$-th input letter
	using a blockwise transduction.
	To handle $\varepsilon$ inputs, we define $\pinput_i$ in two steps.
	Fix $i$ and let
	$A=\{a_1,\ldots,a_m\}$.
	We first obtain an encoded version $u_i$ of
	the $i$-th input from $w$: For every transition that reads $a_j$, we
	replace its transition code with $ab^jab^{m-j}a$. Moreover, for each
	transition that reads $\varepsilon$, we replace the transition code by
	$b^{m+3}$. Using \autoref{lem:blockwise}, this replacement is easily achieved
	using a blockwise transduction. Hence, each possible input in
	$A\cup\{\varepsilon\}$ is encoded using a block from $Y\cup
	\{b^{m+3}\}$, where $Y=\{ab^jab^{m-j}a \mid j\in\{1,\ldots,m\}\}$.

	Suppose we have produced the encoded input $u_i\in (Y\cup \{b^{m+3}\})^*$. In the next
	step, we want to define the word $v_i\in Y^*$, which is obtained from
	$u_i$ by removing each block $b^{m+3}$ from $u_i$. We do this as follows:
	\[ v_i\in Y^* ~\wedge~ v_i\subword u_i ~\wedge~ |v_i|_a = |u_i|_a. \]
	Note that here, we can express $v_i\in Y^*$ because of \autoref{lem:x-star}.
	In the final step, we turn $v_i$ into the input $w_i\in A$ by
	replacing each block $ab^jab^{m-j}a$ with $a_j$ for
	$j\in\{1,\ldots,m\}$. This is just a blockwise transduction and can be
	defined by \autoref{lem:blockwise} because $|B|\ge |A|+1$.
\end{proof}

\begin{rem}
	We do not use this here, but
	\autoref{lem:transductions} also holds without
	the assumption $|B|\ge 3$. Indeed, if $|B|=2$,
	then this would imply $|A_i|=1$ for every $i$.
	Then we can write $A_i=\{a_i\}$ for (not
	necessarily distinct) letters
	$a_1,\ldots,a_k$.  Since $T$ is rational, the
	set of all $(x_1,\ldots,x_k)\in \N^k$ with
	$(a_1^{x_1},\ldots,a_k^{x_k})\in T$ is
	semilinear, and thus $\Sigma_1$-definable in
	$(\N,+,0)$.  It follows from
	the known predicates that $T$ is
	$\Sigma_1$-definable using subwords over
	$\{a_1,\ldots,a_k\}$.
\end{rem}
\subsection*{Step IV: Generator language} Our next ingredient is to express
a particular non-regular language $G$ (and its variant $G_{\#}$):
\begin{align*} G&=\{a^nb^n \mid n\ge 0\}^*, & G_{\#}&=\{a^nb^n\# \mid n\ge 0\}^*. \end{align*}
This will be useful because from $G$,
one can produce all recursively enumerable sets by way of rational
transductions and intersection. 

\begin{lem}\label{lem:replacement}
	The language $\{ab,\#\}^*$ is $\Sigma_1$-definable over $\{a,b,\#\}$.
\end{lem}
\begin{proof}
	Note that
	\[ u\in \{ab,\#\}^* \iff \exists v\in \#^*ab\#^*,~w\in v^* \colon ~u\subword w ~\wedge~ \pi_{a,b}(u)=\pi_{a,b}(w).\]
	Here, the language $\#^*ab\#^*$ can be defined using concatenation.
	Moreover, since
	every word in $\#^*ab\#^*$ is primitive, we express $w\in v^*$ by
	saying $vw=wv$.
\end{proof}

\begin{lem}
	\label{lem:generator}
	Let $\{a,b\}\subseteq A$ and $|A|\ge 3$. 
	The language $G$ is $\Sigma_1$-definable over $A$.
\end{lem}
\begin{proof}
	Suppose $\#\in A\setminus \{a,b\}$. Since $G=\pi_{a,b}(G_{\#})$, it suffices to define $G_{\#}$.
	We can define the language $a^*b^*\#$ as a concatenation of $a^*$, $b^*$, and $\#$.
	The next step is to define the language $K=(a^*b^*\#)^*$. To this end, notice that
	\[ w\in K \iff \exists u\in a^*b^*\#,~v\in u^*\colon~ w\subword v ~\wedge~ |w|_{\#}=|v|_{\#}. \]
	Here, since the words in $a^*b^*\#$ are primitive, we can express $v\in u^*$ by saying $vu=uv$. Thus, we can define $K$.
	Using $K$ and \autoref{lem:replacement}, we can define $G_{\#}$, since
	\[
		w \in G_{\#} \iff w \in K ~\wedge~ \exists v \in \{ab,\#\}^* \colon ~\pi_{a,\#}(w)=\pi_{a,\#}(v) ~\wedge~ \pi_{b,\#}(w)=\pi_{b,\#}(v).\qedhere
	\]
\end{proof}

\subsection*{Step V: Recursively enumerable relations over two letters}
We are now ready to define all recursively enumerable relations over two letters
in $\FOconst{A}$, provided that $|A|\ge 3$.
For two rational transductions $T\subseteq A^*\times B^*$ and $S \subseteq B^* \times C^*$, and a language $L\subseteq A^*$, we denote \emph{application} of $T$ to $L$ as $TL=\{v \in B^* \mid \exists u\in L\colon (u,v)\in T\} \subseteq B^*$, and
we denote \emph{composition} of $S$ and $T$ as $S \circ T = \{(u,w) \mid \exists v \in B^* \colon (u,v) \in T \wedge (v,w) \in S \} \subseteq A^* \times C^*$.
The latter is again a rational transduction (see e.g.\ \cite{Berstel1979}).

\begin{lem}[Hartmanis \& Hopcroft 1970]\label{hartmanis-hopcroft}
 Every recursively enumerable language $L$ can be written as $L=\alpha(T_1
 G_{\#}\cap T_2 G_{\#})$ with a morphism $\alpha$ and rational transductions $T_1,T_2$.
\end{lem}
\begin{proof}
	This follows directly from \cite[Theorem~1]{HartmanisHopcroft1970} and
	the proof of \cite[Theorem~2]{HartmanisHopcroft1970}.
\end{proof}

\noindent Let us briefly sketch the proof of \autoref{hartmanis-hopcroft}.  It essentially
states that every recursively enumerable language can be accepted by a machine
with access to two counters that work in a restricted way.  The two counters
have instructions to \emph{increment}, \emph{decrement}, and \emph{zero test}
(which correspond to the letters $a$, $b$, and $\#$ in $G_{\#}$). The
restriction, which we call ``locally one-reversal''~(L1R) is that in between
two zero tests of some counter, the instructions of that counter must be
\emph{one-reversal}: There is a phase of increments and then a phase of
decrements (in other words: after a decrement, no increments are allowed until
the next zero test).  

To show this, Hartmanis and Hopcroft use the classical fact that every
recursively enumerable language can be accepted by a four counter machine
(without the L1R property). Then, the four counter values $p,q,r,s$ can be
encoded as $2^p 3^q 5^r 7^s$ in a single integer register that can (i)~multiply
with, (ii)~divide by, (iii)~test non-divisibility by the constants $2,3,5,7$.
Such a register, in turn, is easily simulated using two L1R-counters: For
example, to multiply by $f\in\{2,3,5,7\}$, one uses a loop that decrements the
first counter and increments the second by $f$, until the first counter is
zero. The other instructions are similar.
\begin{lem}\label{re-transducer}
	For every recursively enumerable relation $R\subseteq(\{a,b\}^*)^k$,
	there is a rational transduction $T\subseteq (\{a,b\}^*)^{k+2}$ such
	that
	\begin{equation} (w_1,\ldots,w_k)\in R \iff \exists u,v\in G\colon (w_1,\ldots,w_k,u,v)\in T. \label{re-transducer-eq}\end{equation}
\end{lem}
\begin{proof}
	We shall build $T$ out of several other transductions. These will be
	over larger alphabets, but since we merely compose them to obtain $T$,
	this is not an issue.

	A standard fact from computability theory states that a relation is
	recursively enumerable if and only if it is the homomorphic image
	of some recursively enumerable language.
	In particular, there is a recursively enumerable language $L\subseteq B^*$ and
	morphisms $\beta_1,\ldots,\beta_k$ such that $R=\{(\beta_1(w),\ldots,\beta_k(w)) \mid
	w\in L\}$.  By \autoref{hartmanis-hopcroft}, we may write
	$L=\alpha(T_1G_{\#}\cap T_2G_{\#})$ for a morphism $\alpha\colon C^*\to B^*$ and
	rational transductions $T_1,T_2\subseteq \{a,b,\#\}^* \times C^*$.
  
	Notice that if $\gamma\colon \{a,b,\#\}^*\to\{a,b\}^*$ is the morphism
	with $\gamma(a)=a$, $\gamma(b)=b$, and $\gamma(\#)=abab$, then
	$G_{\#}=(a^*b^*\#)^*\cap\gamma^{-1}(G)$.
	Taking the pre-image under a morphism (here $\gamma$) and then intersecting with the regular language (here $(a^*b^*\#)^*$)
	can be performed by a single rational transduction~\cite[Theorem~3.2]{Berstel1979}.
	This means, there is a rational transduction $S\subseteq \{a,b\}^* \times \{a,b,\#\}^*$ with
	$G_{\#}=SG$. Therefore, we can replace $G_{\#}$ in the above expression
  for $L$ and arrive at $L=\alpha\big((T_1(SG)\cap (T_2 (SG)\big) = \alpha\big((T_1 \circ S)G\cap (T_2 \circ S)G\big)$.
  In sum, we observe that	$(w_1,\ldots,w_k)\in R$
	if and only if there exists a $w\in C^*$ with $w\in (T_1 \circ
	S)G$ and $w\in (T_2 \circ S)G$ such that $w_i=\beta_i(\alpha(w))$ for
	$i\in\{1,\ldots,k\}$.  Consider the relation 
	\[ T= \{(\beta_1(\alpha(w)), \ldots, \beta_k(\alpha(w)), u, v) \mid w \in C^*,~ (u,w) \in T_1 \circ S,~ (v,w) \in T_2 \circ S \}.\]
	Note that $T$ is rational: A transducer can guess $w$, letter by
	letter, and on track $i\in\{1,\ldots,k\}$, it outputs the image under
	$\beta_i(\alpha(\cdot))$ of each letter. To compute the output on
	tracks $k+1$ and $k+2$, it simulates transducers for $T_1 \circ S$ and
	$T_2 \circ S$.  Moreover, we have
	$T\subseteq(\{a,b\}^*)^{k+2}$ and our observation implies that
	\eqref{re-transducer-eq} holds.
\end{proof}

\begin{lem} \label{binaryGPRM}
	Let $A$ be an alphabet with $\{a,b\} \subseteq A$ and $|A| \ge 3$.  Then
	every recursively enumerable relation $R \subseteq (\{a,b\}^*)^k$ is
	$\Sigma_1$-definable over $A$.
\end{lem}
\begin{proof}
	Take the rational transduction $T$ as in \autoref{re-transducer}.  Since
	$T\subseteq(\{a,b\}^*)^{k+2}$ and $|A|\ge |\{a,b\}|+1$,
	\autoref{lem:transductions} and \autoref{lem:generator} yield the result.
\end{proof}

\subsection*{Step VI: Arbitrary recursively enumerable relations}
We have seen that if $|A|\ge 3$, then we can define over $A$ every recursively
enumerable relation over two letters. In the proof, we use a third letter as an
auxiliary letter. Our last step is to define all recursively enumerable
relations that can use all letters of $A$ freely. This clearly implies \autoref{thm:main}.
To this end, we observe that every word is determined by its
binary projections.
\begin{lem}\label{lem:binaryprojections}
Let $A$ be an alphabet with $|A|\ge 2$ and let $u,v\in A^*$ such that for every binary alphabet
$B\subseteq A$, we have $\pi_B(u)=\pi_B(v)$. Then $u=v$.
\end{lem}
\begin{proof}
	Towards a contradiction, suppose $u\ne v$. We clearly have $|u|=|v|$.
	Thus, if $w\in A^*$ is the longest common prefix of $u$ and $v$, then
	$u=wau'$ and $v=wbv'$ for some letters $a\ne b$ and words $u',v'\in
	A^*$. But then the words $\pi_{a,b}(u)$ and $\pi_{a,b}(v)$ differ:
	After the common prefix $\pi_{a,b}(w)$, the word $\pi_{a,b}(u)$ continues
	with $a$ and the word $\pi_{a,b}(v)$ continues with $b$.
\end{proof}

We now fix $a,b \in A$ with $a \neq b$.
For any binary alphabet $B \subseteq A$ let $\rho_B \colon A^* \to \{a,b\}^*$ be any morphism with
$\rho_B(B) = \{a,b\}$ and $\rho_B(c) = \varepsilon$ for all $c \in A \setminus B$,
i.e., $\rho_B$ first projects a word over $A$ to $B$ and then renames the letters from $B$ to $\{a,b\}$. Recall that $\binom{|A|}{2}$ is the number of binary alphabets $B\subseteq A$.
We define the encoding function $e \colon A^* \to (\{a,b\}^*)^{\binom{|A|}{2}}$
which maps a word $u \in A^*$ to the tuple consisting of all words $\rho_B(u)$
for all binary alphabets $B \subseteq A$ (in some arbitrary order). Note that $e$ is
injective by \autoref{lem:binaryprojections}.

\begin{lem}\label{lem:encoding}
	If $|A|\ge 3$, then $e \colon A^* \to (\{a,b\}^*)^{\binom{|A|}{2}}$ is $\Sigma_1$-definable over $A$.
\end{lem}
\begin{proof}
	For binary alphabets $B,C\subseteq A$, a map $\sigma\colon B^*\to C^*$
	is called a \emph{binary renaming} if (i)~$\sigma$ is a 
	word morphism and (ii)~$\sigma$ restricted to $B$ is a bijection of $B$
	and $C$. If, in addition, there is a letter $\#\in B\cap C$ such that
	$\sigma(\#)=\#$, then we say that $\sigma$ \emph{fixes a letter}.

	Observe that if we can $\Sigma_1$-define all binary renamings, then the
	encoding function $e$ can be $\Sigma_1$-defined using projections and
	binary renamings. Thus, it remains to define all binary renamings. For
	this, note that every binary renaming can be written as a composition
	of (at most three) binary renamings that each fix some letter. Hence, it
	suffices to define any binary renaming that fixes a letter. Suppose
	$\sigma \colon \{c,\#\}^* \to \{d,\#\}^*$ with $\sigma(c) = d$ and
	$\sigma(\#) = \#$.  Without loss of generality, we assume $c\ne d$. Then $\sigma$ is $\Sigma_1$-definable since
	\[ \sigma(u) = v \iff \exists w\in \{cd,\#\}^*\colon~ u=\pi_{c,\#}(w)~\wedge~ v=\pi_{d,\#}(w). \]
	and $\{cd,\#\}^*$ is definable by \autoref{lem:replacement}.
\end{proof}

\begin{thm} \label{thm:GPRM}
	Let $A$ be an alphabet with $|A| \ge 3$.
	Then every recursively enumerable relation $R \subseteq (A^*)^k$ is $\Sigma_1$-definable in $\FOconst{A}$.
\end{thm}
\begin{proof}
	The encoding function $e$ is clearly computable and injective by \autoref{lem:binaryprojections}.
	Therefore a relation $R \subseteq (A^*)^k$ is recursively enumerable if and only if the image
	\[ e(R)=\{(e(w_1),\ldots,e(w_k)) \mid (w_1,\ldots,w_k)\in R\} \subseteq (\{a,b\}^*)^{k \cdot \binom{|A|}{2}} \]
	is recursively enumerable. This means that $e(R)$ is $\Sigma_1$-definable over $A$ by \autoref{binaryGPRM}. Thus, we can define $R$ as well, since we have
	\[ (w_1,\ldots,w_k)\in R \iff (e(w_1),\ldots,e(w_k))\in e(R), \]
  and the function $e$ is $\Sigma_1$-definable over $A$ by \autoref{lem:encoding}.
\end{proof}

\section{Restricting the signature to a single constant}
\label{sec:single-constant}

In \cite[Remark~3.4]{HalfonSZ17} the authors observe that their undecidability result
for the existential fragment of subword logic with constants still holds, even if
only a finite set of constant symbols is allowed.
More precisely, they show that for any alphabet $A$ with $|A| \geq 2$ there are
finitely many words $w_1, \ldots, w_n \in A^*$ such that the truth problem for the
$\Sigma_1$-fragment over the structure $\FOsomeconst{A}{w_1,\ldots,w_n}$ is undecidable.
Furthermore, they remark that one can strengthen this result even more to only requiring
a single constant $W \in A^*$.
The proof of the latter can be found in the extended version \cite[Theorem~3.5]{HalfonSZ17arxiv}.
Using similar techniques, we can likewise strengthen our $\Sigma_1$-definability result
for alphabets of size at least $3$:

\begin{thm} \label{thm:GPRM-1c}
  Let $A$ be an alphabet with $|A| \ge 3$.
  Then there is a fixed word $W \in A^*$ such that every recursively enumerable relation $R \subseteq (A^*)^k$ is $\Sigma_1$-definable in $\FOsomeconst{A}{W}$.
\end{thm}

\noindent Like in \cite{HalfonSZ17arxiv}, we make use of the fact that any word
of length at least $3$ is uniquely defined by its set of strict subwords:

\begin{lem}[\cite{KarandikarS15}] \label{lem:uniquely-defining-subwords}
  Let $u,v$ be words with $|u| \geq 3$ and $|v| \geq 3$.
  Then $u = v$ if and only if $u {\downarrow} \setminus \{u\} = v {\downarrow} \setminus \{v\}$.
\end{lem}
Note that \autoref{lem:uniquely-defining-subwords} does not hold for words of length $2$ since $ab {\downarrow} \setminus \{ab\} = \{\varepsilon, a, b\} = ba {\downarrow} \setminus \{ba\}$.

\begin{proof}[Proof of \autoref{thm:GPRM-1c}]
  We begin with the case $A = \{a,b,c\}$, i.e.\ $|A| = 3$.
  
  \noindent Recall that concatenation is $\Sigma_1$-definable in the structure $\FOconst{A}$. 
  Let $w_1,$ $\ldots,$ $w_n$ be the constant symbols that appear in the formula defining
  the concatenation relation.
  We choose $W = a^{m+1}b^{m+2}c^{m+3}$, where $m = \max_{1 \leq i \leq n}|w_i|$ is
  the maximal length among these constants.
  Since every word is a concatenation of letters, it now suffices to show that
  the letters $a, b, c$ and words $w_1, \ldots, w_n$ are all $\Sigma_1$-definable
  over $\FOsomeconst{A}{W}$.
  
  In the following we use $u \ssubword v$ as a shorthand for $u \subword v \wedge u \neq v$.
  The formula
  \[
    v_{3m + 5} \ssubword W \wedge \bigwedge_{i=0}^{3m+4} v_i \ssubword v_{i+1}
  \]
  defines a sequence of subwords $v_0, \ldots, v_{3m+5}$ of $W$ with $|v_i| = i$.
  In particular, we have $v_1 \in A$.
  If we repeat the same formula for two more sets of variables $u_0,\ldots,u_{3m+5}$, $t_0, \ldots, t_{3m+5}$, and additionally require
  \[
    v_1 \neq u_1 \wedge v_1 \neq t_1 \wedge u_1 \neq t_1,
  \]
  then we have defined $a$, $b$, and $c$, but only up to renaming of letters.
  To ensure $v_1 = c$ we use the following formula:
  \[
     u_1 \not\subword v_{m+3}' \wedge t_1 \not\subword v_{m+3}'
    \wedge v_{m+3}' \subword W \wedge \bigwedge_{i=0}^{m+2} v_i' \ssubword v_{i+1}'.
  \]
  It defines a sequence of $v_0', \ldots v_{m+3}'$ of subwords of $W$, among which the
  letters $u_1$ and $t_1$ do not occur.
  Therefore this sequence is comprised of words in $v_1^*$, and by choice of $W$
  a sequence of this length cannot exist for $v_1 = a$ or $v_1 = b$.
  Observe that also $|v_i'| = i$, which means that we have now additionally defined
  the words $c^0 = \varepsilon$ to $c^{m+3}$.
  Using two similar formulas with $m+3$ and $m+2$ variables, respectively, we can now
  define the words $b^0$ to $b^{m+2}$ and $a^0$ to $a^{m+1}$ as well.
  
  Observe that for any word $w$ and any letter $a'$, $|w|_{a'} = \ell$ is equivalent
  to $a^\ell \subword w \wedge a^{\ell+1} \not\subword w$.
  Using this fact we can fix the number of occurrences of each letter
  $a'$, and therefore also fix the total length of a word.
  We continue by defining the remaining words of length $\leq 2$, which are
  $ab, ba, ac, ca, bc, cb$.
  The formula
  \[
    a \subword s_{01} \wedge aa \not\subword s_{01} \wedge b \subword s_{01} \wedge bb \not\subword s_{01} \wedge c \not\subword s_{01} \wedge s_{01} \subword W
  \]
  defines the word $s_{01} = ab$.
  By copying this formula for a new free variable $s_{10}$ and replacing the last
  conjunct by $s_{10} \not\subword W$, we can likewise define $s_{10} = ba$.
  Similarly, we can also distinguish $ac$ from $ca$, as well as $bc$ from $cb$,
  since in both cases one of them is a subword of $W$ while the other is not.
  
  Finally we inductively define all constants up to length $m$. Let $w$ be a word
  of length $3 \leq |w| \leq m$.
  Then by induction hypothesis we have defined all constants up to length $|w| - 1$.
  Furthermore for every $a' \in A$ we have defined the constants $a'^{|w|_{a'}}$ and
  $a'^{|w|_{a'} + 1}$, since $|w|_{a'} \leq |w| \leq m$.
  By \autoref{lem:uniquely-defining-subwords} the word $w$ is uniquely defined by its set of
  strict subwords.
  Therefore the following formula defines $s = w$:
  \begin{align*}
    &a^{|w|_a} \subword s \wedge a^{|w|_a + 1} \not\subword s
    \wedge b^{|w|_b} \subword s \wedge b^{|w|_b + 1} \not\subword s
    \wedge c^{|w|_c} \subword s \wedge c^{|w|_c + 1} \not\subword s \\
    &\wedge \quad \bigwedge_{|w'| \leq |w|, w' \subword w} w' \subword s \quad
    \wedge \quad \bigwedge_{|w'| \leq |w|, w' \not\subword w} w' \not\subword s.
  \end{align*}
  
  Since $a,b,c$ and $w_1, \dots, w_n$ have at most length $m$, we have successfully
  defined all of the required constants.
  To conclude the case $|A| = 3$, we add existential quantifiers for all the variables
  representing constants that we do not use outside of these auxiliary formulas.
  Since we only ever used existential quantification, we have shown definability
  in the $\Sigma_1$-fragment of $\FOsomeconst{A}{W}$.
  
  The case $A = \{a_1,\ldots,a_k\}$ for $k \geq 4$ is very similar.
  We define the number $m$ in the same way as before and choose
  $W = a_1^{m + 1} \cdots a_k^{m + k}$.
  Then we can again define all constants $a_i^0$ to $a_i^{m+1}$ for every letter $a_i$.
  We also proceed to define all remaining words of length $2$, in the same way as
  in the previous case.
  Finally, \autoref{lem:uniquely-defining-subwords} and the fact
  $|w|_{a_i} = \ell \iff a_i^\ell \subword w \wedge a_i^{\ell+1} \not\subword w$
  allow us to inductively define all constants $w$ up to length $m$, like before.
\end{proof}

\begin{rem}
  In the proof of \autoref{thm:GPRM-1c} we show that there is a number $m \in \N$ such
  that for the alphabet $A = \{a_1,\ldots,a_k\}$ with $k \geq 4$ the word
  $a_1^{m + 1} \cdots a_k^{m + k}$ is a valid choice for the constant symbol $W$.
  Moreover, the proof still works for any number $m'>m$ and
  $W = a_1^{m' + 1} \cdots a_k^{m' + k}$.
  Therefore the theorem does not just hold for one fixed word $W$, but an infinite family
  of words $W_{m'}$.
\end{rem}

\section{Further consequences}\label{consequences}

In this section, we prove \autoref{level-by-level-constants}, \autoref{level-by-level-pure} and \autoref{sigma1-pure}.  When working
with higher levels ($\Sigma_i^0$ for $i\ge 2$) of the arithmetic hierarchy, it
will be convenient to use a slightly different definition than the one using
oracle Turing machines: \cite[Theorem~35.1]{kozen2006theory} implies that for
$i\ge 1$, a relation $R\subseteq (A^*)^k$ belongs to $\Sigma_{i+1}^0$ if and
only if it can be written as $R=\pi((A^*)^{k+\ell}\setminus S)$, where
$S\subseteq(A^*)^{k+\ell}$ is a relation in $\Sigma_{i}^0$ and $\pi\colon
(A^*)^{k+\ell}\to (A^*)^k$ is the projection to the first $k$ coordinates.

\begin{proof}[Proof of \autoref{level-by-level-constants}]
	It is immediate that every predicate definable in the
	$\Sigma_i$-fragment of $\FOconst{A}$ belongs to $\Sigma_i^0$, because
	the subword relation is recursively enumerable. We show the converse
	using induction on $i$, such that \autoref{thm:main} is the base case.  
	
	Now suppose that every relation in $\Sigma_i^0$ is definable in the
	$\Sigma_i$-fragment of $\FOconst{A}$ and consider a relation
	$R\subseteq (A^*)^k$ in $\Sigma_{i+1}^0$. Then we can write
	$R=\pi((A^*)^{k+\ell}\setminus S)$ for some $\ell\ge 0$, where
	$\pi\colon (A^*)^{k+\ell}\to (A^*)^k$ is the projection to the first
	$k$ coordinates, and $S\subseteq (A^*)^{k+\ell}$ is a relation in
	$\Sigma^0_i$. By induction, $S$ is definable by a $\Sigma_i$-formula
	$\varphi$ over $\FOconst{A}$. By negating $\varphi$ and moving all
	negations inwards, we obtain a $\Pi_i$-formula $\psi$ that defines
	$(A^*)^{k+\ell}\setminus S$. Finally, adding existential quantifiers
	for the variables corresponding to the last $\ell$ coordinates yields a
	$\Sigma_{i+1}$-formula for $R=\pi((A^*)^{k+\ell}\setminus S)$.
\end{proof}

Note that if instead of \autoref{thm:main} we use \autoref{thm:GPRM-1c}
as the base case in the induction above, it follows that there exists
a word  $W \in A^*$ such that \autoref{level-by-level-constants} also
holds for the structure $\FOsomeconst{A}{W}$ instead of $\FOconst{A}$.
This together with \autoref{thm:GPRM-1c} yields \autoref{rem:single-constant}.

Finally, we look at the expressive power of the pure logic $\FOpure{A}$.
We start by proving \autoref{level-by-level-pure}, which characterizes relations definable
in the $\Sigma_i$-fragment of $\FOpure{A}$ for $i \ge 2$.

\begin{proof}[Proof of \autoref{level-by-level-pure}]
	Clearly, every relation definable with a
	$\Sigma_i$-formula over $\FOpure{A}$ must be
	automorphism-invariant and must define a relation
	in $\Sigma_i^0$.

	Conversely, consider an automorphism-invariant relation $R\subseteq
	(A^*)^k$ in $\Sigma_i^0$. Then $R$ is definable using a
	$\Sigma_i$-formula $\varphi$ with free variables $x_1,\ldots,x_k$ over
	$\FOconst{A}$ by \autoref{level-by-level-constants}. Let
	$w_1,\ldots,w_\ell$ be the constants occurring in $\varphi$. From
	$\varphi$, we construct the $\Sigma_i$-formula $\varphi'$ over
	$\FOpure{A}$, by replacing each occurrence of $w_j$ by a fresh variable
	$y_j$.
	
	It was shown in \cite[Sections 4.1 and 4.2]{KarandikarS15} that from
	the tuple $(w_1,\ldots,w_\ell)\in (A^*)^\ell$, one can construct a
	$\Sigma_2$-formula $\psi$ with free variables $y_1,\ldots,y_\ell$ over
	$\FOpure{A}$ such that $\psi(u_1,\ldots,u_\ell)$ is true if and only if
	there exists an automorphism of $\FOpure{A}$ mapping $u_j$ to $w_j$ for
	each $j$. We claim that the formula $\chi=\exists y_1,\ldots,y_\ell\colon
	\psi\wedge \varphi'$ defines the set $R$. Since $\psi$ belongs to $\Sigma_2$
	and thus $\chi$ belongs to $\Sigma_i$, this implies \autoref{level-by-level-pure}.

	Clearly, every $(v_1,\ldots,v_k)\in R$ satisfies $\chi$. Moreover, if
	$\chi(v_1,\ldots,v_k)$, then there are $u_1,\ldots,u_\ell\in A^*$ with
	$\varphi'(v_1,\ldots,v_k,u_1,\ldots,u_\ell)$ and an automorphism $\alpha$
	mapping $u_j$ to $w_j$ for each $j$. Since $\alpha$ is an
	automorphism, the formula $\varphi'$ is also satisfied on the tuple
	$(\alpha(v_1),\ldots,\alpha(v_k),\alpha(u_1),\ldots,\alpha(u_\ell))=(\alpha(v_1),\ldots,\alpha(v_k),w_1,\ldots,w_\ell)$
	and thus we have $(\alpha(v_1),\ldots,\alpha(v_k))\in R$.  
	Since $R$ is automorphism-invariant, this implies
	$(v_1,\ldots,v_k)\in R$.
\end{proof}

\noindent The expressiveness of the existential fragment of $\FOpure{A}$ is not well understood.
Partial results are summarized in \autoref{sigma1-pure}.

\begin{proof}[Proof of \autoref{sigma1-pure}]
	Take a recursively enumerable, but undecidable subset $S\subseteq\N$.
	Fix a letter $a \in A$ and define the unary language $L=\{a^n \mid n\in S\}$.
	By \cite[Theorem~3.5]{HalfonSZ17arxiv} there exists a word $W \in A^*$
	and a $\Sigma_1$-formula $\varphi(x)$ over $\FOsomeconst{A}{W}$ which defines $L$.
	Consider the formula $\varphi'$ in the $\Sigma_1$-fragment of $\FOpure{A}$
	obtained by replacing each occurrence of $W$ by a fresh variable~$y$.
  Then $(v,W)$ satisfies $\varphi'$ if and only if
	$v \in L$. Thus, $\varphi'$ defines an undecidable relation.
	
	For the second statement, we claim that every language $L\subseteq A^*$
	that is $\Sigma_1$-definable in $\FOpure{A}$ satisfies
	$A^*LA^*\subseteq L$. Hence, many automorphism-invariant regular languages
	such as $\bigcup_{a\in A} a^*$ are not definable.
	Note that for $a\in A$ and $u,v\in A^*$, we
	have $u\subword v$ if and only if $au\subword av$.   Thus, every
	$\Sigma_0$-definable relation $R\subseteq (A^*)^k$ satisfies
	$(w_1,\ldots,w_k)\in R$ if and only if $(aw_1,\ldots,aw_k)\in R$.
	Symmetrically, $(w_1,\ldots,w_k)\in R$ is equivalent to
	$(w_1a,\ldots,w_ka)\in R$. As a projection of a $\Sigma_0$-definable
	relation, $L$ thus satisfies $A^*LA^*\subseteq L$.
\end{proof}

\noindent \autoref{sigma1-pure} raises the question whether there are undecidable $\Sigma_1$-definable \emph{languages} in $\FOpure{A}$,
which we leave as an open problem. In fact, all examples of $\Sigma_1$-definable languages that we have constructed are regular.
While each individual $\Sigma_1$-definable language could be decidable,
the following observation implies that the membership problem for $\Sigma_1$-definable languages is \emph{not} decidable,
if the formula is part of the input.

\begin{obs}
	\label{thm:magicConstantUndecidable}
	Let $|A| \ge 2$. There exists a word $W \in A^*$ such that the following problem is undecidable:
	Given a $\Sigma_1$-formula $\varphi(x)$ over $\FOpure{A}$, does $W$ satisfy $\varphi(x)$?
\end{obs}
\begin{proof}
	By \cite[Theorem~3.5]{HalfonSZ17arxiv} there exists a word $W \in A^*$ so that the \emph{truth problem}
	for $\Sigma_1$-formulas over $\FOsomeconst{A}{W}$ is undecidable, which asks whether a given a sentence
	$\varphi$ (formula without free variables) is true.
	In other words, testing whether $W$ satisfies the formula obtained from $\varphi$
	by replacing each occurrence of $W$ by a fresh free variable is undecidable.
\end{proof}

\section{Conclusion}\label{sec:conclusion}
We have shown how to define all recursively enumerable relations in the existential fragment of the subword order with constants for each alphabet $A$ with $|A| \geq 3$. 
If $|A|=1$, then the relations definable in $\FOconst{A}$ correspond to
relations over $\N$ definable in $(\N,\le)$ with constants. Hence, this case is
very well understood: This structure admits quantifier
elimination~\cite[Theorem 2.2(b)]{peladeau1992logically}, which implies that
the $\Sigma_1$-fragment is expressively complete and also that a subset of
$A^*$ is only definable if it is finite or co-finite. In particular,
\autoref{thm:main} does not hold for $|A|=1$.

We leave open whether \autoref{thm:main} still holds over a binary alphabet.
If this is the case, then we expect that substantially new techniques are required.
In order to express non-trivial relations over two letters, our proof often
uses a third letter as a separator and marker for ``synchronization points''
in subword embeddings.

\bibliographystyle{alphaurl}
\bibliography{references.bib}

\appendix

\end{document}